\documentclass[submission,copyright,creativecommons]{eptcs}
\pdfoutput=1
\providecommand{\event}{Linearity \& TLLA 2020} 
\usepackage{breakurl}             
\usepackage{underscore}           

\usepackage{stmaryrd}
\usepackage{mathpartir}
\usepackage{amssymb}
\usepackage{cmll}
\usepackage{xcolor}
\usepackage{paralist}
\usepackage{amsmath}
\usepackage{amsthm}
\usepackage{mathrsfs}
\usepackage{mathtools}
\usepackage{multicol}
\usepackage{tabularx}
\usepackage{tikz-cd}
\usepackage[conor]{agda}

\DeclareFontFamily{U}{cal}{}
\DeclareFontShape{U}{cal}{m}{n}{<->cmsy10}{}
\DeclareSymbolFont{rcal}{U}{cal}{m}{n}
\DeclareSymbolFontAlphabet{\mathcal}{rcal}

\usepackage{thmtools}
\declaretheorem[numberwithin=section]{theorem}

\declaretheorem[numberlike=theorem]{proposition}
\declaretheorem[numberlike=theorem]{lemma}
\declaretheorem[numberlike=theorem]{corollary}
\declaretheorem[numberlike=theorem]{example}
\declaretheorem[numberlike=theorem]{definition}
\declaretheorem[numberlike=theorem]{remark}

\def\newelims{1}
\def\debruijn{1}
\def\multnotation{1}
\ifx\newelims\undefined
  \def\newelims{0}
\fi
\ifx\multnotation\undefined
  \def\multnotation{0}
\fi

\newcommand{\bind}[2]{%
  \if\debruijn0%
  {#2}\{{#1}\}%
  \else%
  {#2}%
  \fi%
}
\newcommand{\ctx}[2]{%
  \if\multnotation0%
  #1^{\resctx{#2}}%
  \else%
  \resctx{#2}#1%
  \fi
}
\newcommand{\ctxvar}[3]{%
  \if\debruijn0%
  \if\multnotation0%
  #1 \stackrel{\rescomment{#3}}: #2%
  \else%
  \rescomment{#3}#1 : #2%
  \fi%
  \else%
  \rescomment{#3}#2%
  \fi%
}
\definecolor{res}{HTML}{008000}
\definecolor{resperm}{HTML}{008000}
\newcommand{\rescomment}[1]{{\color{res}#1}}
\newcommand{\rescommentperm}[1]{{\color{resperm}#1}}
\newcommand{\resctx}[1]{\rescomment{\mathcal{#1}}}
\newcommand{\resctxperm}[1]{\rescommentperm{\mathcal{#1}}}

\newcommand{\base}[0]{\iota}

\newcommand{\fun}[2]{#1 \multimap #2}
\newcommand{\lam}[2]{%
  \if\debruijn0%
  \lambda #1.~#2%
  \else%
  \lambda #2%
  \fi%
}
\newcommand{\app}[2]{#1~#2}

\newcommand{\excl}[2]{\oc \rescomment{#1} #2}
\newcommand{\bang}[1]{\left[#1\right]}
\newcommand{\bm}[4]{%
  \if\newelims0%
  \operatorname{bm}_{#1}(#2, \bind{#3}{#4})%
  \else%
  \if\debruijn0%
  \mathrm{let}~\bang{#3} = #2~\mathrm{in}~#4%
  \else%
  \mathrm{let}~\bang{-} = #2~\mathrm{in}~#4%
  \fi%
  \fi%
}

\newcommand{\tensorOne}[0]{1}
\newcommand{\unit}[0]{(\mathbin{_\otimes})}
\newcommand{\del}[3]{\if\newelims0%
\operatorname{del}_{#1}(#2, #3)%
\else%
\mathrm{let}~\unit = #2~\mathrm{in}~#3%
\fi}

\newcommand{\tensor}[2]{#1 \otimes #2}
\newcommand{\ten}[2]{(#1 \mathbin{_{\otimes}} #2)}
\newcommand{\prm}[5]{%
  \if\newelims0%
  \operatorname{pm}_{#1}(#2, \bind{#3, #4}{#5})%
  \else%
  \if\debruijn0%
  \mathrm{let}~\ten{#3}{#4} = #2~\mathrm{in}~#5%
  \else%
  \mathrm{let}~\ten{-}{-} = #2~\mathrm{in}~#5%
  \fi%
  \fi%
}

\newcommand{\withTOne}[0]{\top}
\newcommand{\eat}[0]{(\mathbin{_{\with}})}

\newcommand{\withT}[2]{#1 \with #2}
\newcommand{\wth}[2]{(#1 \mathbin{_\with} #2)}
\newcommand{\proj}[2]{\operatorname{proj}_{#1} #2}

\newcommand{\sumTZero}[0]{0}
\newcommand{\exf}[2]{\operatorname{ex-falso} #2}

\newcommand{\sumT}[2]{#1 \oplus #2}
\newcommand{\inj}[2]{\operatorname{inj}_{#1} #2}
\newcommand{\cse}[6]{%
  \if\newelims0%
  \operatorname{case}_{#1}(#2, \bind{#3}{#4}, \bind{#5}{#6})%
  \else%
  \if\debruijn0%
  \mathrm{case}~#2~\mathrm{of}~\inj{L}{#3} \mapsto #4
                            ~;~ \inj{R}{#5} \mapsto #6%
  \else%
  \mathrm{case}~#2~\mathrm{of}~\inj{L}{-} \mapsto #4 ~;~ \inj{R}{-} \mapsto #6
  \fi%
  \fi%
}

\newcommand{\lvar}{\mathrel{\mathrlap{\sqsupset}{\mathord{-}}}}

\newcommand{\subres}{\trianglelefteq}

\makeatletter
\newcommand{\subst}[2][]{\ext@arrow 0359\Rightarrowfill@{#1}{#2}}
\makeatother

\newenvironment{eqns}{\begin{array}{r@{\hspace{0.3em}}c@{\hspace{0.3em}}l}}{\end{array}}

\newcommand{\mat}[1]{\mathbf{#1}}
\newcommand{\vct}[1]{\mathbf{#1}}

\newcommand{\name}{\ensuremath{\lambda \resctxperm R}}

\DeclareMathOperator\kit{Kit}
\newcommand{\kitrel}{\mathbin{\blacklozenge}}

\DeclareMathOperator\id{id}
\DeclareMathOperator\inl{inl}
\DeclareMathOperator\inr{inr}
\DeclareMathOperator\Idx{Idx}

\newcommand{\zero}{\ensuremath{\rescomment 0}}
\newcommand{\linear}{\ensuremath{\rescomment 1}}
\newcommand{\unrestricted}{\ensuremath{\rescomment \omega}}
\newcommand{\instDILL}{\rescomment{01\omega}}

\newcommand{\unused}{\ensuremath{\rescomment{0}}}
\newcommand{\true}{\ensuremath{\rescomment{1}}}
\newcommand{\valid}{\ensuremath{\rescomment{\square}}}
\newcommand{\instPD}{\ensuremath{\rescomment{01\square}}}

\title{A Linear Algebra Approach to Linear Metatheory}
\author{
James Wood\thanks{James Wood is supported by an EPSRC Studentship.}
\institute{University of Strathclyde\\ Glasgow, United Kingdom}
\email{james.wood.100@strath.ac.uk}
\and
Robert Atkey
\institute{University of Strathclyde\\ Glasgow, United Kingdom}
\email{robert.atkey@strath.ac.uk}
}
\def\titlerunning{A Linear Algebra Approach to Linear Metatheory}
\def\authorrunning{J. Wood \& R. Atkey}
\begin{document}
\maketitle

\begin{abstract}
  Linear typed $\lambda$-calculi are more delicate than their simply
  typed siblings when it comes to metatheoretic results like
  preservation of typing under renaming and substitution. Tracking the
  usage of variables in contexts places more constraints on how
  variables may be renamed or substituted. We present a methodology
  based on linear algebra over semirings, extending McBride's
  \emph{kits and traversals} approach for the metatheory of syntax with
  binding to linear usage-annotated terms. Our approach is readily
  formalisable, and we have done so in Agda.
\end{abstract}

\section{Introduction}

The basic metatheoretic results for typed $\lambda$-calculi, such as
preservation of typing under renaming, weakening, substitution and so
on, are crucial but quite boring to prove. In calculi with
substructural typing disciplines and modalities, it can also be quite
easy to break these properties \cite{wadler91use,BentonBPH93}. It is desirable
therefore to use a proof assistant to prove these properties. This has
the double benefit of both confidence in the results and in focusing
on the essential properties required to obtain them.

Mechanisation of the metatheory of substructural $\lambda$-calculi has
not received the same level of attention as intuitionistic
typing. ``Straightforward'' translations from paper presentations to
formal presentations make metatheory difficult, due to incompatibilities
between the standard de Bruijn representation of binding and the
splitting of contexts. For formalisations of linear sequent calculi,
sticking to the paper presentation using lists and permutations is
common \cite{power99,XavierORN18,laurent18}, but explicit permutations
make the resulting encodings difficult to use. Multisets for contexts
are more convenient \cite{ChaudhuriLR19}, but do not work well for
Curry-Howard uses, as noted by Laurent. For natural deduction, Allais
\cite{allais:LIPIcs:2018:10049} uses an I/O model to track usage of
variables, Rouvoet et al.\ \cite{RPKV20} use a co-de Bruijn
representation to distribute variables between subterms, and Crary
uses mutually defined typing and linearity judgements with HOAS
\cite{crary10}.


In this paper, we adapt the generic \emph{kits and traversals}
technique for proving admissibility of renaming and substitution due
to McBride \cite{rensub05} to a linear typed $\lambda$-calculus where
variables are annotated with values from a skew semiring denoting
those variables' \emph{usage} by terms. Our calculus, \name{}, is a
prototypical example of a linear ``quantitative'' or ``coeffect''
calculus in the style of
\cite{reed10distance,BrunelGMZ14,GhicaS14,PetricekOM14,Granule18}. The
key advantages of \name{} over the formalisations listed above are
that the shape of typing contexts is maintained, so de Bruijn indices
behave the same as in non-substructural calculi, and by selecting
different semirings, we obtain from \name{} well known systems,
including Barber's Dual Intuitionistic Linear Logic \cite{Barber1996}
and Pfenning and Davies' S4 modal type theory \cite{judgmental}.

McBride's kits and traversals technique isolates properties required
to form binding-respecting traversals of simply typed $\lambda$-terms,
so that renaming and substitution arise as specific
instantiations. Benton, Hur, Kennedy, and McBride \cite{bhkm12}
implement the technique in Coq and extend it to polymorphic
terms. Allais \emph{et al.}~\cite{AACMM20} generalise to a wider class
of syntax with binding and show that more general notions of
\emph{semantics} can be handled.
Using methods like these can reduce the effort required to develop a new
calculus and its metatheory.

To adapt kits and traversals to linear usage-annotated terms requires
us to not only respect the binding structure, but to also respect the
usage annotations. For instance, the usages associated with a term
being substituted in must be correctly distributed across all the
occurrences of that term in the result. To aid us in tracking usages,
we employ the linear algebra of vectors and matrices
induced by the skew semiring we are using. Usage annotations on
contexts are vectors, usage-preserving maps of contexts are matrices,
and the linearity properties of the maps induced by matrices are
exactly the lemmas we need for showing that traversals (and hence
renaming, subusaging, and substitution) preserve typing and usages.



The paper proceeds as follows:
\begin{itemize}
  \item In \autoref{sec:algebra}, we specify our requirements on the set of
    annotations that will track usage of variables. A consequence of our
    formalisation is that we learn that we only need \emph{skew} semirings, a
    weaker structure than the partially ordered semirings usually used.
  \item In \autoref{sec:syntax}, we use these annotations to define the system
    $\name$ in an intrinsically typed style.
  \item In \autoref{sec:general}, we show some characteristics of $\name$ under
    certain general conditions on the skew semiring of usage annotations.
  \item In \autoref{sec:translation}, we show that, via choice of semiring, we
    can embed Barber's Dual Intuitionistic Linear Logic \cite{Barber1996} and
    Pfenning and Davies' modal calculus \cite{judgmental}.
  \item In \autoref{sec:metatheory}, we prove that \name{} admits renaming,
    subusaging, and substitution by our extension of McBride's kits and
    traversals technique.
  \item We conclude in \autoref{sec:conclusion} with some directions for future
    work.
\end{itemize}

\autoref{sec:general} and \autoref{sec:translation} can be read as extended
examples of the $\name$ syntax.
Those two sections logically come after \autoref{sec:metatheory}, but we think
it helpful for readers unfamiliar with semiring-annotated calculi to read them
before proceeding to \autoref{sec:metatheory}.
Conversely, a reader familiar with semiring-annotated calculi who is primarily
interested in our metatheoretic methods may skip \autoref{sec:general} and
\autoref{sec:translation} without issue.

The Agda formalisation of this work can be found at
\url{https://github.com/laMudri/generic-lr/tree/lin/tlla-submission-2021/src/Specific}.
It contains our formalisation of vectors and
matrices (approx.\ 790 lines) and the definition of \name{} and proofs of
renaming and substitution (approx.\ 530 lines).

\section{Skew Semirings}\label{sec:algebra}

We shall use skew semirings where other authors have previously used partially ordered
semirings (see, for example, the Granule definition of a \emph{resource algebra}
\cite{Granule18}).
Elements of a skew semiring are used as \emph{usage annotations}, and describe
\emph{how} values are used in a program.
In the syntax for \name{}, each assumption will have a usage annotation,
describing how that assumption can be used in the derivation.
Addition describes how to combine multiple usages of an assumption, and
multiplication describes the action our graded $\oc$-modality can have.
The ordering describes the specificness of annotations.
If $p \subres q$, $p$ can be the annotation for a variable wherever $q$ can be.
We can read this relation as ``$\textrm{supply} \subres \textrm{demand}$'' ---
given a variable annotated $p$, we can coerce to treat it as if it has the
desired annotation $q$.

Skew semirings are a generalisation of partially ordered semirings, which are in turn a
generalisation of commutative semirings.
As such, readers unfamiliar with the more general structures may wish to think
in terms of the more specific structures.
Our formalisation was essential for noticing and sticking to this level of
generality.

\begin{definition}
  A \emph{(left) skew monoid} is a structure $(\mathbf R, \subres, 1, *)$ such
  that $(\mathbf R, \subres)$ forms a partial order, $*$ is monotonic with
  respect to $\subres$, and the following laws hold (with $x * y$ henceforth
  being written as $xy$).
  \begin{mathpar}
    1x \subres x
    \and x \subres x1
    \and (xy)z \subres x(yz)
  \end{mathpar}
\end{definition}

\begin{remark}
  A commutative skew monoid is just a partially ordered commutative monoid.
\end{remark}

Skew-monoidal categories are due to Szlach\'anyi \cite{skew}, and the notion
introduced here of a skew monoid is a decategorification of the notion of
skew-monoidal category.

\begin{definition}
  A \emph{(left) skew semiring} is a structure
  $(\mathbf R, \subres, 0, +, 1, *)$ such that $(\mathbf R, \subres)$ forms a
  partial order, $+$ and $*$ are monotonic with respect to $\subres$,
  $(\mathbf R, 0, +)$ forms a commutative monoid, $(\mathbf R, \subres, 1, *)$
  forms a skew monoid, and we have the following distributivity laws.
  \begin{mathpar}
    0z \subres 0
    \and (x + y)z \subres xz + yz
    \and 0 \subres x0
    \and xy + xz \subres x(y + z)
  \end{mathpar}
\end{definition}

\begin{example}
  In light of the above remark, most ``skew'' semirings are actually
  just partially ordered semirings. An example that yields a system
  equivalent to Barber's DILL is the
  $\zero \triangleright \unrestricted \triangleleft \linear$
  semiring of ``unused'', ``unrestricted'', and ``linear'', respectively.
  See \cite{Granule18} for more examples.
\end{example}

We will only speak of \emph{left} skew semirings, and thus generally
omit the word ``left''.  A mnemonic for (left) skew semirings is
``multiplication respects operators on the left from left to right,
and respects operators on the right from right to left''.  One may
also describe multiplication as ``respecting'' and ``corespecting''
operators on the left and right, respectively.

From a skew semiring $\mathbf R$, we form finite vectors, which we
notate as $\mathbf R^n$, and matrices, which we notate as
$\mathbf{R}^{m\times n}$. In Agda, we represent vectors in
$\mathbf R^n$ as functions $\Idx n \to \mathbf{R}$, where
$\Idx n$ is the type of valid indexes in an $n$-tuple, and
matrices in $\mathbf{R}^{m\times n}$ as functions
$\Idx m \to \Idx n \to \mathbf{R}$.  Whereas elements
of $\mathbf R$ describe how individual \emph{variables} are used,
elements of $\mathbf R^n$ describe how all of the variables in an
$n$-length \emph{context} are used. We call such vectors \emph{usage
  contexts}, and take them to be row vectors. Matrices in
$\mathbf{R}^{m\times n}$ will be used to describe how usage contexts
are transformed by renaming and substitution in
\autoref{sec:metatheory}. We define $\subres$, $0$ and $+$ on vectors
and matrices pointwise. Basis vectors $\langle i \rvert$ (used to
represent usage contexts for individual variables), identity matrices
$\mat I$, matrix multiplication $*$, and matrix reindexing
${-}_{{-}\times{-}}$ are defined as follows:
\begin{mathpar}
  \begin{matrix*}[l]
    \langle {-} \rvert : \Idx n \to \mathbf{R}^n \\
    \langle i \rvert_j \coloneqq
    \begin{cases}
      1, & \textrm{if }i = j \\
      0, & \textrm{otherwise} \\
    \end{cases}
  \end{matrix*}
  \and
  \begin{matrix*}[l]
    \mat I : \mathbf R^{m \times m} \\
    \mat I_{ij} \coloneqq \langle i \rvert_j
  \end{matrix*}
  \and
  \begin{matrix*}[l]
    * : \mathbf R^{m \times n} \times \mathbf R^{n \times o} \to \mathbf R^{m \times o} \\
    (MN)_{ik} \coloneqq \sum_j M_{ij}N_{jk}
  \end{matrix*}
  \and
  \begin{matrix*}[l]
    {-}_{{-}\times{-}} : \mathbf R^{m' \times n'}
    \times (\Idx m \to \Idx m')
    \times (\Idx n \to \Idx n')
    \to \mathbf R^{m \times n} \\
    \left(M_{f \times g}\right)_{i,j} \coloneqq M_{f\,i,g\,j}
  \end{matrix*}
\end{mathpar}

We define vector-matrix multiplication by treating vectors
as $1$-height matrices.
If $i : \Idx m$ and $j : \Idx n$, then $\inl i$ and $\inr j$ are both of type
$\Idx(m + n)$.
In prose, there will always be canonical choices for $m$ and $n$, whereas in the
mechanisation, we work with the free pointed magma over the $1$-element set as
opposed to the free monoid over the $1$-element set (i.e., the natural numbers),
so it is unambiguous where there is a sum of dimensionalities.

\section{Syntax}\label{sec:syntax}

\begin{figure}
  \begin{displaymath}
    A,B,C \Coloneqq \base \mid \fun A B \mid \tensorOne \mid \tensor A B
    \mid \sumTZero \mid \sumT A B \mid \withTOne \mid \withT A B \mid \excl{r} A
  \end{displaymath}
  \caption{Types of \name{}}
  \label{fig:types}
\end{figure}

\begin{figure}
  \centering
  \begin{tabular}{ll}
    $\Gamma \ni A$          & type of plain (non-usage-checked) variables \\
    $\ctx\Gamma R \lvar A$  & type of usage-checked variables \\
    $\ctx\Gamma R \vdash A$ & type of usage-checked terms
  \end{tabular}
  \caption{Judgement forms of \name{}}
  \label{fig:judgements}
\end{figure}

\begin{figure}
  \begin{mathpar}
    \inferrule*[right=var]
    {x : \ctx{\Gamma}{R} \lvar A}
    {x : \ctx{\Gamma}{R} \vdash A}
    \and
    \inferrule*[right=$\fun{}{}$-E]
    {M : \ctx{\Gamma}{P} \vdash \fun{A}{B}
      \\ N : \ctx{\Gamma}{Q} \vdash A
      \\ \resctx R \subres \resctx P + \resctx Q
    }
    {\app{M}{N} : \ctx{\Gamma}{R} \vdash B}
    \and
    \inferrule*[right=$\fun{}{}$-I]
    {\bind{x}M : \ctx{\Gamma}{R}, \ctxvar{x}{A}{1} \vdash B}
    {\lam{x}{\bind{x}M} : \ctx{\Gamma}{R} \vdash \fun{A}{B}}

    \and

    \inferrule*[right=$\tensorOne$-E]
    {M : \ctx{\Gamma}{P} \vdash \tensorOne{}
      \\ N : \ctx{\Gamma}{Q} \vdash C
      \\ \resctx R \subres \resctx P + \resctx Q
    }
    {\del{C}{M}{N} : \ctx{\Gamma}{R} \vdash C}
    \and
    \inferrule*[right=$\tensorOne$-I]
    {\resctx R \subres \rescomment{\vct 0}}
    {\unit{} : \ctx{\Gamma}{R} \vdash \tensorOne{}}
    \and
    \inferrule*[right=$\tensor{}{}$-E]
    {M : \ctx{\Gamma}{P} \vdash \tensor{A}{B}
      \\ \bind{x,y}N : \ctx{\Gamma}{Q}, \ctxvar{x}{A}{1}, \ctxvar{y}{B}{1}
      \vdash C
      \\ \resctx R \subres \resctx P + \resctx Q
    }
    {\prm{C}{M}{x}{y}{\bind{x,y}N} : \ctx{\Gamma}{R} \vdash C}
    \and
    \inferrule*[right=$\tensor{}{}$-I]
    {M : \ctx{\Gamma}{P} \vdash A
      \\ N : \ctx{\Gamma}{Q} \vdash B
      \\ \resctx R \subres \resctx P + \resctx Q
    }
    {\ten{M}{N} : \ctx{\Gamma}{R} \vdash \tensor{A}{B}}

    \and

    \inferrule*[right=$\sumTZero$-E]
    {M : \ctx{\Gamma}{P} \vdash \sumTZero{}
      \\ \resctx R \subres \resctx P + \resctx Q
    }
    {\exf{C}{M} : \ctx{\Gamma}{R} \vdash C}
    \and
    \inferrule*[right=$\sumT{}{}$-E]
    {M : \ctx{\Gamma}{P} \vdash \sumT{A}{B}
      \\ \bind{x}N : \ctx{\Gamma}{Q}, \ctxvar{x}{A}{1} \vdash C
      \\ \bind{y}O : \ctx{\Gamma}{Q}, \ctxvar{y}{B}{1} \vdash C
      \\ \resctx R \subres \resctx P + \resctx Q
    }
    {\cse{C}{M}{x}{\bind{x}N}{y}{\bind{y}O} : \ctx{\Gamma}{R} \vdash C}
    \and
    \inferrule*[right=$\sumT{}{}$-Il]
    {M : \ctx{\Gamma}{R} \vdash A}
    {\inj{L}{M} : \ctx{\Gamma}{R} \vdash \sumT{A}{B}}
    \and
    \inferrule*[right=$\sumT{}{}$-Ir]
    {M : \ctx{\Gamma}{R} \vdash B}
    {\inj{R}{M} : \ctx{\Gamma}{R} \vdash \sumT{A}{B}}

    \and

    \inferrule*[right=$\withTOne$-I]
    { }
    {\eat{} : \ctx{\Gamma}{R} \vdash \withTOne}
    \and
    \inferrule*[right=$\withT{}{}$-El]
    {M : \ctx{\Gamma}{R} \vdash \withT{A}{B}}
    {\proj{L}{M} : \ctx{\Gamma}{R} \vdash A}
    \and
    \inferrule*[right=$\withT{}{}$-Er]
    {M : \ctx{\Gamma}{R} \vdash \withT{A}{B}}
    {\proj{R}{M} : \ctx{\Gamma}{R} \vdash B}
    \and
    \inferrule*[right=$\withT{}{}$-I]
    {M : \ctx{\Gamma}{R} \vdash A
      \\ N : \ctx{\Gamma}{R} \vdash B
    }
    {\wth{M}{N} : \ctx{\Gamma}{R} \vdash \withT{A}{B}}

    \and

    \inferrule*[right=$\excl{r}{}$-E]
    {M : \ctx{\Gamma}{P} \vdash \excl{r}{A}
      \\ \bind{x}N : \ctx{\Gamma}{Q}, \ctxvar{x}{A}{r} \vdash C
      \\ \resctx R \subres \resctx P + \resctx Q
    }
    {\bm{C}{M}{x}{\bind{x}N} : \ctx{\Gamma}{R} \vdash C}
    \and
    \inferrule*[right=$\excl{r}{}$-I]
    {M : \ctx{\Gamma}{P} \vdash A
      \\ \resctx R \subres \rescomment r\resctx P
    }
    {\bang{M} : \ctx{\Gamma}{R} \vdash \excl{r}{A}}
  \end{mathpar}
  \caption{Typing rules of \name{}}
  \label{fig:rules}
\end{figure}

We present the syntax of \name{} as an \emph{intrinsically} typed
syntax, as it is in our Agda formalisation. Intrinsic typing means
that we define well typed terms as inhabitants of an inductive family
$\ctx\Gamma{R} \vdash A$ indexed by typing contexts $\Gamma$, usage
contexts $\resctx{R}$, and types $A$. Typing contexts are lists of
types. Usage contexts $\resctx{R}$ are vectors of elements of some
fixed skew semiring $\mathbf R$, with the same number of elements as
the typing context they are paired with. To highlight how usage
annotations are used in the syntax, we write all elements of
$\mathbf R$, and vectors and matrices thereof, in \rescomment{green}.



The types of \name{} are given in \autoref{fig:types}.
We have a base type $\base$, function types $\fun A B$, tensor product types
$\tensor A B$ with unit $\tensorOne$, sum types $\sumT A B$ with unit
$\sumTZero$, ``with'' product types $\withT A B$ with unit $\withTOne$, and an
exponential $\excl{r} A$ indexed by a usage $\rescomment{r}$.

We distinguish between \emph{plain} variables, values of type
$\Gamma \ni A$, and \emph{usage-checked} variables, values of type
$\ctx{\Gamma}{R} \lvar A$.
A plain variable is an index into a context with a specified type, while a
usage-checked variable of type $\ctx{\Gamma}{R} \lvar A$ is a plain variable
$i : \Gamma \ni A$ together with a proof that
$\resctx R \subres \langle i \rvert$.
Expanding the vector notation, the latter
condition says that the selected variable $i$ must have a usage
annotation $\subres \rescomment 1$ in $\resctx{R}$, while all other
variables must have a usage annotation $\subres \rescomment 0$. We
will sometimes silently cast between the types $\Idx m$ and
$\Gamma \ni A$, particularly when using the reindexing operation
${-}_{{-}\times{-}}$.





The constructors for our intrinsically typed terms are presented in \autoref{fig:rules}.
In keeping with our intrinsic typing methodology, terms of \name{} are presented as constructors of the inductive family $\ctx\Gamma{R} \vdash A$, hence the notation $M : \ctx\Gamma{R} \vdash A$ instead of the more usual $\ctx\Gamma{R} \vdash M : A$. 
Our Agda formalisation uses de Bruijn indices to represent variables, but we have annotated the rules with variable names for ease of reading. 
Ignoring the \rescomment{usages}, the typing rules all look like their
simply typed counterparts; the only difference between the $\otimes$
and $\with$ products being their presentation in terms of pattern
matching and projections, respectively.
Thus the addition of usage contexts and constraints on them refines the usual simple typing to be usage constrained.
For instance, in the \TirName{$\otimes$-I} rule, the usage context $\resctx R$ on the conclusion is constrained to be able to supply the sum $\resctx P + \resctx Q$ of the usage contexts of the premises.
If we instantiate $\mathbf R$ to be the
$\zero \triangleright \unrestricted \triangleleft \linear$ semiring, then we
obtain a system that is equivalent to Barber's DILL \cite{Barber1996}, as we will see in \autoref{sec:translation}.

For the purposes of our metatheoretical results in \autoref{sec:metatheory}, the precise rules chosen here are
not too important.
The salient point is that contexts can only be manipulated in specific ways.
Typing contexts can only be modified by context extensions, i.e., binding new
variables.
Usage contexts can correspondingly be extended, but also can be
\emph{linearly} split.
A usage context $\resctx R$ can be split into zero pieces via constraint
$\resctx R \subres \rescomment{\vct 0}$, into two pieces, $\resctx P$ and
$\resctx Q$, via constraint $\resctx R \subres \resctx P + \resctx Q$, and an
$\rescomment r$-scaled down piece $\resctx P$ via constraint
$\resctx R \subres \rescomment r\resctx P$.

However, the precise set of rules that we have chosen will be important in \autoref{sec:translation}, as they
correspond closely to the rules of Dual Intuitionistic Linear Logic and the
modal calculus of Pfenning and Davies \cite{Barber1996,judgmental}.
In fact, there are several similar calculi in the literature which do not
embed intuitionistic linear logic, and thus do not translate in the same way as
\name{}.
For example, the system of Abel and Bernardy \cite{AbelBernardy2020}
essentially replaces \TirName{$\tensor{}{}$-E} with the following stronger
rule (in our notation, modified to include subusaging).
This eliminator allows one to derive
$\fun{\excl r{(\tensor A B)}}{\tensor{\excl r A}{\excl r B}}$ for any
$\rescomment r$, $A$, and $B$, whereas linear logic has no such tautology.

\[
  \inferrule*[right=$\tensor{}{}$-E$'$]
  {M : \ctx{\Gamma}{P} \vdash \tensor{A}{B}
    \\ \bind{x,y}N : \ctx{\Gamma}{Q}, \ctxvar{x}{A}{q}, \ctxvar{y}{B}{q}
    \vdash C
    \\ \resctx R \subres \rescomment q\resctx P + \resctx Q
  }
  {\prm{C}{\rescomment qM}{x}{y}{\bind{x,y}N} : \ctx{\Gamma}{R} \vdash C}
\]

In \autoref{sec:metatheory}, we will show the admissibility of, amongst others,
the rules shown in \autoref{fig:lemmas}.
The rules in \autoref{fig:lemmas} will be required for the proofs in
\autoref{sec:general} and \autoref{sec:translation}.
Using these rules, we can also derive the following fact, which demonstrates
a linear form of \emph{cut}.

\begin{lemma}\label{lem:turnstile-derivation}
  If we can derive $\ctxvar x A 1 \vdash B$, then we also know that from
  $\ctx\Gamma R \vdash A$ we can derive $\ctx\Gamma R \vdash B$.
\end{lemma}
\begin{proof}
  Assuming the two hypotheses, we can make the following derivation.
  \[
    \inferrule*[right=SingleSubst]
    {
      \ctx\Gamma R \vdash A
      \\
      \inferrule*[right=Weak]
      {\ctxvar x A 1 \vdash B}
      {\ctx\Gamma{\vct 0}, \ctxvar x A 1 \vdash B}
      \\
      \inferrule*{ }
      {\rescomment 1\resctx R + \vct{\rescomment 0} \subres \resctx R}
    }
    {\ctx\Gamma R \vdash B}
  \]
\end{proof}

\begin{figure}
  \begin{mathpar}
    \inferrule*[right=SingleSubst]
    {M : \ctx\Gamma P \vdash A
      \\ N : \ctx\Gamma Q, \ctxvar x A r \vdash B
      \\ \resctx R \subres \rescomment r \resctx P + \resctx Q
    }
    {N\{M\} : \ctx\Gamma R \vdash B}
    \and
    \inferrule*[right=Subuse]
    {M : \ctx\Gamma Q \vdash A \\ \resctx P \subres \resctx Q}
    {M : \ctx\Gamma P \vdash A}
    \and
    \inferrule*[right=Weak]
    {\ctx\Gamma P \vdash A}
    {\ctx\Gamma P, \ctx\Delta{\vct0} \vdash A}
  \end{mathpar}
  \caption{Admissible rules}
  \label{fig:lemmas}
\end{figure}

\section{Intuitionistic and Modal Instantiations}\label{sec:general}

As we will show in \autoref{sec:dill}, $\name$ can be instantiated with
a semiring that makes the system linear in the sense of DILL.
However, with different choices of semiring, $\name$ can become a calculus
satisfying more structural rules.

\begin{lemma}
  When instantiated with the 1-element (skew) semiring
  $\{\rescomment{\bullet}\}$, $\name$ becomes a variant of intuitionistic
  simply typed $\lambda$-calculus.
\end{lemma}
\begin{proof}
  With only one possible usage annotation, usage contexts do not contain any
  information.
  Because $\rescomment{\bullet} \subres \rescomment{\bullet}$, all
  $\subres$-constraints are satisfied.
  By inspection, when usage contexts and constraints are ignored, the typing
  rules of \autoref{fig:rules} become those of an intuitionistic
  $\lambda$-calculus.
\end{proof}

An important class of skew semirings is those for which $0$ is the top
element of the $\subres$-order and $+$ acts as a meet (minimum).
While $\name$ instantiated this way admits full weakening and
contraction, the usage annotations can still play a part via the
$\excl{r}{}$ modality.  We will see an example of such an
instantiation in \autoref{sec:pd} when we embed Pfenning and Davies' S4
modal $\lambda$-calculus.

\begin{lemma}\label{lem:top-meet}
  If \name{} is instantiated at a skew semiring such that $0$ is top and $+$ is
  meet (with respect to the subusaging order), then $\withTOne$ and
  $\tensorOne$ are interderivable, and $\withT A B$ and $\tensor A B$ are
  interderivable.
\end{lemma}
\begin{proof}
  By \autoref{lem:turnstile-derivation}, following derivations suffice.
  \begin{mathpar}\footnotesize
    \inferrule*[right=$\withTOne$-I]
    { }
    {\ctxvar x \tensorOne 1 \vdash \withTOne}

    \and

    \inferrule*[right=$\tensorOne$-I]
    {\inferrule*{ }{\rescomment 1 \subres \rescomment 0}}
    {\ctxvar x \withTOne 1 \vdash \tensorOne}

    \and

    \inferrule*[right=$\tensor{}{}$-E]
    {
      \inferrule*[right=var]
      {
        \inferrule*{ }{(\rescomment 1) \subres (\rescomment 1)}
      }
      {\ctxvar x{(\tensor A B)}1 \vdash \tensor A B}
      \\
      \inferrule*[right=$\withT{}{}$-I]
      {
        \inferrule*[right=var]
        {
          \inferrule*{ }
          {(\rescomment 0, \rescomment 1, \rescomment 1) \subres
            (\rescomment 0, \rescomment 1, \rescomment 0)}
        }
        {\ctxvar x{(\tensor A B)}0, \ctxvar x A 1, \ctxvar y B 1 \vdash A}
        \\
        \inferrule*[right=var]
        {
          \inferrule*{ }
          {(\rescomment 0, \rescomment 1, \rescomment 1) \subres
            (\rescomment 0, \rescomment 0, \rescomment 1)}
        }
        {\ctxvar x{(\tensor A B)}0, \ctxvar x A 1, \ctxvar y B 1 \vdash B}
      }
      {\ctxvar x{(\tensor A B)}0, \ctxvar x A 1, \ctxvar y B 1
        \vdash \withT A B}
      \\
      \inferrule*{ }{(\rescomment 1) + (\rescomment 0) \subres (\rescomment 1)}
    }
    {\ctxvar x{(\tensor A B)}1 \vdash \withT A B}

    \and

    \inferrule*[right=$\tensor{}{}$-I]
    {
      \inferrule*[right=$\withT{}{}$-El]
      {
        \inferrule*[right=var]
        {\inferrule*{ }{(\rescomment 1) \subres (\rescomment 1)}}
        {\ctxvar x {(\withT A B)}1 \vdash \withT A B}
      }
      {\ctxvar x {(\withT A B)}1 \vdash A}
      \\
      \inferrule*[right=$\withT{}{}$-Er]
      {
        \inferrule*[right=var]
        {\inferrule*{ }{(\rescomment 1) \subres (\rescomment 1)}}
        {\ctxvar x {(\withT A B)}1 \vdash \withT A B}
      }
      {\ctxvar x {(\withT A B)}1 \vdash B}
      \\
      \inferrule*{ }{(\rescomment 1) \subres (\rescomment 1) + (\rescomment 1)}
    }
    {\ctxvar x {(\withT A B)}1 \vdash \tensor A B}
  \end{mathpar}
\end{proof}

\section{Translation to and from Existing Systems}
\label{sec:translation}

A motivating reason to consider the system presented in this paper is that
instances of it correspond to previously studied systems.
In this section, we present translations from \name{} to Dual Intuitionistic
Linear Logic \cite{Barber1996} and the modal system of Pfenning and Davies
\cite{judgmental}, and vice versa.
We cannot prove that the translations form an equivalence, because we have not
written down an equational theory for \name{}, but we expect this to be easy
enough to do.

\subsection{Dual Intuitionistic Linear Logic}\label{sec:dill}
Dual Intuitionistic Linear Logic is a particular formulation of intuitionistic
linear logic \cite{Barber1996}.
Its key feature, which simplifies the metatheory of linear logic, is the use of
separate contexts for linear and intuitionistic free variables.
Here we show that DILL is a fragment of the instantiation of \name{} at the
linearity semiring $\{\zero, \linear, \unrestricted\}$.

\begin{definition}
  Let $\instDILL$ denote the following semiring on the partially ordered set
  $\{\linear \triangleright \zero \triangleleft \unrestricted\}$.
  \begin{multicols}{2}
    \begin{itemize}
      \item $0 \coloneqq \zero$
      \item
        \begin{tabular}{c|ccc}
          $+$ & \zero & \linear & \unrestricted \\ \hline
          \zero & \zero & \linear & \unrestricted \\
          \linear & \linear & \unrestricted & \unrestricted \\
          \unrestricted & \unrestricted & \unrestricted & \unrestricted \\
        \end{tabular}
      \item $1 \coloneqq \linear$
      \item
        \begin{tabular}{c|ccc}
          $*$ & \zero & \linear & \unrestricted \\ \hline
          \zero & \zero & \zero & \zero \\
          \linear & \zero & \linear & \unrestricted \\
          \unrestricted & \zero & \unrestricted & \unrestricted \\
        \end{tabular}
    \end{itemize}
  \end{multicols}
\end{definition}

The types of DILL are the same as the types of \name, except for the
restriction of $\excl{r}{}$ to just $\excl{\unrestricted}{}$.
We will write the latter simply as $\oc$ when it appears in DILL\@.
We add sums and with-products to the calculus of \cite{Barber1996}, with the
obvious rules (stated fully in \autoref{fig:dill}).
These additive type formers present no additional difficulty to the translation.

\begin{figure}
  \begin{mathpar}
    \inferrule*[right=Int-Ax]{ }
    {\Gamma, A; \cdot \vdash A}
    \and
    \inferrule*[right=Lin-Ax]{ }
    {\Gamma; A \vdash A}
    \and
    \inferrule*[right=$I$-I]{ }
    {\Gamma; \cdot \vdash I}
    \and
    \inferrule*[right=$I$-E]
    {\Gamma; \Delta_1 \vdash I \\ \Gamma; \Delta_2 \vdash A}
    {\Gamma; \Delta_1, \Delta_2 \vdash A}
    \and
    \inferrule*[right=$\otimes$-I]
    {\Gamma; \Delta_1 \vdash A \\ \Gamma, \Delta_2 \vdash B}
    {\Gamma; \Delta_1, \Delta_2 \vdash A \otimes B}
    \and
    \inferrule*[right=$\otimes$-E]
    {\Gamma; \Delta_1 \vdash A \otimes B \\ \Gamma; \Delta_2, A, B \vdash C}
    {\Gamma; \Delta_1, \Delta_2 \vdash C}
    \and
    \inferrule*[right=$\multimap$-I]
    {\Gamma; \Delta, A \vdash B}
    {\Gamma; \Delta \vdash A \multimap B}
    \and
    \inferrule*[right=$\multimap$-E]
    {\Gamma; \Delta_1 \vdash A \multimap B \\ \Gamma; \Delta_2 \vdash A}
    {\Gamma; \Delta_1, \Delta_2 \vdash B}
    \and
    \inferrule*[right=$\oc$-I]
    {\Gamma; \cdot \vdash A}
    {\Gamma; \cdot \vdash \oc A}
    \and
    \inferrule*[right=$\oc$-E]
    {\Gamma; \Delta_1 \vdash \oc A \\ \Gamma, A; \Delta_2 \vdash B}
    {\Gamma; \Delta_1, \Delta_2 \vdash B}
    \and
    \inferrule*[right=$\top$-I]{ }
    {\Gamma; \Delta \vdash \top}
    \and
    \inferrule*[right=$\with$-I]
    {\Gamma; \Delta \vdash A \\ \Gamma, \Delta \vdash B}
    {\Gamma; \Delta \vdash A \with B}
    \and
    \inferrule*[right=$\with$-E$_i$]
    {\Gamma; \Delta \vdash A_0 \with A_1}
    {\Gamma; \Delta \vdash A_i}
    \and
    \inferrule*[right=$0$-E]
    {\Gamma; \Delta_1 \vdash 0}
    {\Gamma; \Delta_1, \Delta_2 \vdash A}
    \and
    \inferrule*[right=$\oplus$-I$_i$]
    {\Gamma; \Delta \vdash A_i}
    {\Gamma; \Delta \vdash A_0 \oplus A_1}
    \and
    \inferrule*[right=$\oplus$-E]
    {
      \Gamma; \Delta_1 \vdash A \oplus B \\
      \Gamma; \Delta_2, A \vdash C \\
      \Gamma; \Delta_2, B \vdash C
    }
    {\Gamma; \Delta_1, \Delta_2 \vdash C}
  \end{mathpar}
  \caption{The rules of DILL, extended with additive connectives}
  \label{fig:dill}
\end{figure}

\begin{figure}
  \begin{mathpar}
    \begin{eqns}
      \mathrm{DILL} &\hookrightarrow& \name_\instDILL \\
      Y &\mapsto& \base_Y \\
      I &\mapsto& \tensorOne \\
      A \otimes B &\mapsto& \tensor{A}{B} \\
      A \multimap B &\mapsto& \fun{A}{B} \\
      \oc A &\mapsto& \excl{\unrestricted}{A} \\
      0 &\mapsto& \sumTZero \\
      A \oplus B &\mapsto& \sumT A B \\
      \top &\mapsto& \withTOne \\
      A \with B &\mapsto& \withT A B
    \end{eqns}
    \and
    \begin{eqns}
      \mathrm{PD} &\hookrightarrow& \name_\instPD \\
      Y &\mapsto& \base_Y \\
      \top &\mapsto& \tensorOne \\
      A \wedge B &\mapsto& \withT{A}{B} \\
      A \supset B &\mapsto& \fun{A}{B} \\
      \Box A &\mapsto& \excl{\valid}{A} \\
      \bot &\mapsto& \sumTZero \\
      A \vee B &\mapsto& \sumT{A}{B}
    \end{eqns}
  \end{mathpar}
  \label{fig:dill-pd}
  \caption{Embedding of DILL and PD types into \name}
\end{figure}

\begin{proposition}[DILL $\to$ \name]
  Given a DILL derivation of $\Gamma; \Delta \vdash A$, we can produce a
  $\name_{\instDILL}$ derivation of
  $\ctx{\Gamma}{\vct \unrestricted}, \ctx{\Delta}{\vct 1} \vdash A$.
\end{proposition}
\begin{proof}
  By induction on the derivation.
  We have $\unrestricted \subres \zero$, which allows us to discard
  intuitionistic variables at the var rules, and both
  $\linear \subres \linear$ and $\unrestricted \subres \linear$, which allow
  us to use both linear and intuitionistic variables.

  Weakening is used when splitting linear variables between two premises.
  For example, \TirName{$\otimes$-I} in DILL is as follows.
  \[
    \inferrule*[right=$\otimes$-I]
    {\Gamma; \Delta_t \vdash t : A \\ \Gamma; \Delta_u \vdash u : B}
    {\Gamma; \Delta_t, \Delta_u \vdash t \otimes u : A \otimes B}
  \]
  From this, our new derivation is as follows.
  \[
    \inferrule*[right=$\tensor{}{}$-I]
    {
      \inferrule*[right=Weak]
      {\mathit{ih}_t \\\\
        \ctx{\Gamma}{\vct\unrestricted}, \ctx{\Delta_t}{\vct\linear}
        \vdash M_t : A}
      {\ctx{\Gamma}{\vct\unrestricted}, \ctx{\Delta_t}{\vct\linear},
        \ctx{\Delta_u}{\vct\zero}
        \vdash M_t : A}
      \\
      \inferrule*[right=Weak]
      {\mathit{ih}_u \\\\
        \ctx{\Gamma}{\vct\unrestricted}, \ctx{\Delta_u}{\vct\linear}
        \vdash M_u : A}
      {\ctx{\Gamma}{\vct\unrestricted}, \ctx{\Delta_t}{\vct\zero},
        \ctx{\Delta_u}{\vct\linear}
        \vdash M_u : A}
    }
    {\ctx{\Gamma}{\vct\unrestricted}, \ctx{\Delta_t}{\vct\linear},
      \ctx{\Delta_u}{\vct\linear}
      \vdash \ten{M_t}{M_u} : \tensor{A}{B}}
  \]
\end{proof}

When translating from \name{} to DILL, we first coerce the \name{} derivation
to be in a form easily amenable to translation into DILL.
An example of a \name{} derivation with no direct translation into DILL is the
following.
In DILL terms, the intuitionistic variable of the conclusion becomes a linear
variable in the premises.
Such a move is admissible in DILL, but does not come naturally.

\[
  \inferrule*[right=$\tensor{}{}$-I]
  {
    \inferrule*[right=var]{ }
    {\ctxvar{x}{A}{\linear} \vdash A}
    \\
    \inferrule*[right=var]{ }
    {\ctxvar{x}{A}{\linear} \vdash A}
    \\
    \unrestricted \subres \linear + \linear
  }
  {\ctxvar{x}{A}{\unrestricted} \vdash \tensor A A}
\]

To avoid such situations, and therefore manipulations on DILL derivations, we
show that all $\name_{\instDILL}$ derivations can be made in \emph{bottom-up}
style.
In bottom-up style, the algebraic facts we make use of are dictated by making
most general choices based on the conclusions of rules.
Bottom-up style corresponds to a (non-deterministic) form of
\emph{usage checking}, and the following lemma can be understood as saying
that that form of usage checking is sufficiently general.

\begin{definition}
  A derivation is said to be \emph{$\instDILL$-bottom-up} if only the following
  facts about addition and multiplication are used, and all proofs of
  inequalities not at leaves are by reflexivity (i.e, not using the facts that
  $\unrestricted \subres \zero$ and $\unrestricted \subres \linear$).

  \makebox[\textwidth][s]{
    \begin{tabular}{c|ccc}
      $+$ & \zero & \linear & \unrestricted \\ \hline
      \zero & \zero & \linear & - \\
      \linear & \linear & - & - \\
      \unrestricted & - & - & \unrestricted \\
    \end{tabular}
    \begin{tabular}{c|ccc}
      $*$ & \zero & \linear & \unrestricted \\ \hline
      \zero & - & - & \zero \\
      \linear & \zero & \linear & \unrestricted \\
      \unrestricted & \zero & - & \unrestricted \\
    \end{tabular}
  }
\end{definition}

Bottom-up style enforces that whenever we split a context into two (for
example, in the rule \TirName{$\tensor{}{}$-I}) all unused variables in the
conclusion stay unused in the premises, intuitionistic variables stay
intuitionistic, and linear variables go either left or right.
Multiplication is only used in the rule \TirName{$\excl{r}{}$-I}, at which point
both the result and left argument are available.
Here, the bottom-up style enforces that linear variables never appear in the
premise of \TirName{$\excl{\unrestricted}{}$-I}.

\begin{lemma}
  Every $\name_{\instDILL}$ derivation can be translated into a bottom-up
  $\name_{\instDILL}$ derivation.
\end{lemma}
\begin{proof}
  By induction on the shape of the derivation.
  When we come across a non-bottom-up use of addition, it must be that the
  corresponding variable in the conclusion has annotation $\unrestricted$.
  By subusaging, we can give this variable annotation $\unrestricted$ in
  the premises too, before translating the subderivations to bottom-up
  style.
  A similar argument applies to uses of multiplication, remembering that both
  the left argument and result are fixed.
\end{proof}

\begin{proposition}[\name{} $\to$ DILL]
  Given a $\name_{\instDILL}$ derivation of
  $\ctx{\Gamma}{\vct\unrestricted}, \ctx{\Delta}{\vct\linear},
  \ctx{\Theta}{\vct\zero} \vdash A$ which does not contain type formers
  $\excl 0$ and $\excl 1$, we can produce a DILL derivation of
  $\Gamma; \Delta \vdash A$.
\end{proposition}
\begin{proof}
  By induction on the derivation having been translated to bottom-up form.

  In the case of \TirName{var}, all of the unused variables have annotation
  greater than $\zero$, i.e., $\zero$ or $\unrestricted$.
  Those annotated $\zero$ are absent from the DILL derivation, and those
  annotated $\unrestricted$ are in the intuitionistic context.
  The used variable is annotated either $\linear$ or $\unrestricted$.
  In the first case, we use \TirName{Lin-Ax}, and in the second case,
  \TirName{Int-Ax}.

  All binding of variables in \name{} maps directly onto DILL.

  Because we translated to bottom-up form, additions, as seen in, for example,
  the \TirName{$\tensor{}{}$-I} rule, can be handled straightforwardly.
  Any intuitionistic variables in the conclusion correspond to intuitionistic
  variables in both premises.
  Any linear variables in the conclusion correspond to a linear variable in
  exactly one of the premises, and is absent in the other premise.

  The only remaining rule is \TirName{$\excl{r}{}$-I}, of which we only cover
  \TirName{$\excl{\unrestricted}{}$-I} (the other two targeting types not found
  in DILL).
  In this case, we know that every variable in the conclusion is annotated
  either $\zero$ or $\unrestricted$, and every variable in the premise is
  annotated the same way.
  This corresponds exactly to the restrictions of DILL's \TirName{$\oc$-I}.
\end{proof}

\subsection{Pfenning Davies}\label{sec:pd}

The translation to and from the modal system of Pfenning and Davies
\cite{judgmental} (henceforth \emph{PD}) is similar to the translation to and
from DILL.
We present our variant of PD, again adding some common connectives, in
\autoref{fig:pd}.
The main difference is the algebra at which \name{} is instantiated.

\begin{figure}
  \begin{mathpar}
    \inferrule*[right=hyp]{ }
    {\Gamma; \Delta, A\;\mathit{true} \vdash A\;\mathit{true}}
    \and
    \inferrule*[right=hyp*]{ }
    {\Gamma, A\;\mathit{valid}; \Delta \vdash A\;\mathit{true}}
    \and
    \inferrule*[right=$\supset$I]
    {\Gamma; \Delta, A\;\mathit{true} \vdash B\;\mathit{true}}
    {\Gamma; \Delta \vdash A \supset B\;\mathit{true}}
    \and
    \inferrule*[right=$\supset$E]
    {
      \Gamma; \Delta \vdash A \supset B\;\mathit{true} \\
      \Gamma; \Delta \vdash A\;\mathit{true}
    }
    {\Gamma; \Delta \vdash B\;\mathit{true}}
    \and
    \inferrule*[right=$\Box$I]
    {\Gamma; \cdot \vdash A\;\mathit{true}}
    {\Gamma; \Delta \vdash \Box A\;\mathit{true}}
    \and
    \inferrule*[right=$\Box$-E]
    {
      \Gamma; \Delta \vdash \Box A\;\mathit{true} \\
      \Gamma, A\;\mathit{valid}; \Delta \vdash B\;\mathit{true}
    }
    {\Gamma; \Delta \vdash B\;\mathit{true}}
    \and
    \inferrule*[right=$\top$-I]{ }
    {\Gamma; \Delta \vdash \top\;\mathit{true}}
    \and
    \inferrule*[right=$\wedge$-I]
    {
      \Gamma; \Delta \vdash A\;\mathit{true} \\
      \Gamma, \Delta \vdash B\;\mathit{true}
    }
    {\Gamma; \Delta \vdash A \wedge B\;\mathit{true}}
    \and
    \inferrule*[right=$\wedge$-E$_i$]
    {\Gamma; \Delta \vdash A_0 \wedge A_1\;\mathit{true}}
    {\Gamma; \Delta \vdash A_i\;\mathit{true}}
    \and
    \inferrule*[right=$\bot$-E]
    {\Gamma; \Delta \vdash \bot\;\mathit{true}}
    {\Gamma; \Delta \vdash A\;\mathit{true}}
    \and
    \inferrule*[right=$\vee$-I$_i$]
    {\Gamma; \Delta \vdash A_i\;\mathit{true}}
    {\Gamma; \Delta \vdash A_0 \vee A_1\;\mathit{true}}
    \and
    \inferrule*[right=$\vee$-E]
    {
      \Gamma; \Delta \vdash A \vee B\;\mathit{true} \\
      \Gamma; \Delta, A \vdash C\;\mathit{true} \\
      \Gamma; \Delta, B \vdash C\;\mathit{true}
    }
    {\Gamma; \Delta \vdash C\;\mathit{true}}
  \end{mathpar}
  \caption{The rules of PD, extended with several standard connectives}
  \label{fig:pd}
\end{figure}

\begin{definition}
  Let $\instPD$ denote the following semiring on the partially ordered set
  $\{\valid \triangleleft \true \triangleleft \unused\}$.
  \begin{multicols}{2}
    \begin{itemize}
      \item $0 \coloneqq \unused$.
      \item $+$ is the meet ($\wedge$) according to the subusaging order.
      \item $1 \coloneqq \true$.
      \item
        \begin{tabular}{c|ccc}
          $*$ & \unused & \true & \valid \\ \hline
          \unused & \unused & \unused & \unused \\
          \true & \unused & \true & \valid \\
          \valid & \unused & \valid & \valid \\
        \end{tabular}
    \end{itemize}
  \end{multicols}
\end{definition}

The $\unused$ annotation plays only a formal role in this example.
Meanwhile, $\true$ and $\valid$ correspond to the judgement forms
$\mathit{true}$ and $\mathit{valid}$ from PD\@.
Addition being the meet makes it idempotent.
Furthermore, it gives us that $\true + \valid = \valid$ --- if somewhere we
require an assumption to be true, and elsewhere require it to be valid, then
ultimately it must be valid (from which we can deduce that it is true).
Multiplication is designed to make $\excl{\valid}{}$ act like PD's $\Box$.
In particular, $\valid * \valid = \valid$ says that the valid assumptions are
available before and after \TirName{$\excl{\valid}{}$-I}, whereas
$\valid * \true = \valid$ says that valid assumptions in the conclusion can be
weakened to true assumptions in the premise.
The latter fact does not appear in PD, and will be excluded from
\emph{bottom-up} derivations.

To keep our notation consistent with that of DILL, we swap the roles of
$\Gamma$ and $\Delta$ in PD compared to what they were in the original paper.
Thus, our PD judgements are of the form $\Gamma; \Delta \vdash A~\mathit{true}$,
where $\Gamma$ contains valid assumptions and $\Delta$ contains true
assumptions.

\begin{proposition}[PD $\to$ \name]
  Given a PD derivation of $\Gamma; \Delta \vdash t : A~\mathit{true}$, we can
  produce a $\name_{\instPD}$ derivation of
  $\ctx{\Gamma}{\vct\valid}, \ctx{\Delta}{\vct\true} \vdash A$.
\end{proposition}
\begin{proof}
  By induction on the PD derivation.
  Most PD rules have direct $\name$ counterparts, noting that variables of any
  annotation can be discarded and duplicated because we have both
  $\rescomment r \subres \rescomment 0$ and
  $\rescomment r \subres \rescomment r + \rescomment r$ for all
  $\rescomment r$.

  Care must be taken with the \TirName{$\Box$I} rule.
  We have, from the induction hypothesis, a $\name$ derivation of
  $\ctx\Gamma{\vct\Box} \vdash A$.
  By \TirName{$\excl\valid{}$-I}, we have
  $\ctx\Gamma{\vct\Box} \vdash \excl\valid A$.
  To get the desired conclusion, we must use \TirName{Weak} to get
  $\ctx\Gamma{\vct\Box}, \ctx\Delta{\vct\unused} \vdash \excl\valid A$, and
  then \TirName{Subuse} on the variables we just introduced (noting that
  $\true \subres \unused$) to get
  $\ctx\Gamma{\vct\Box}, \ctx\Delta{\vct\true} \vdash \excl\valid A$.
\end{proof}

For translating from $\name_{\instPD}$ to PD, we introduce a similar notion of
\emph{bottom-up} derivations as we did for DILL\@.
Every $\name_{\instPD}$ derivation can be translated into bottom-up style, and
then be directly translated into PD.

\begin{definition}
  A derivation is said to be \emph{$\instPD$-bottom-up} if only the following
  facts about addition and multiplication are used, and all proofs of
  inequalities not at leaves are by reflexivity.

  \makebox[\textwidth][s]{
    \begin{tabular}{c|ccc}
      $+$ & \unused & \true & \valid \\ \hline
      \unused & \unused & - & - \\
      \true & - & \true & - \\
      \valid & - & - & \valid \\
    \end{tabular}
    \begin{tabular}{c|ccc}
      $*$ & \unused & \true & \valid \\ \hline
      \unused & - & - & \unused \\
      \true & \unused & \true & \valid \\
      \valid & \unused & - & \valid \\
    \end{tabular}
  }
\end{definition}

\begin{lemma}
  Every $\name_{\instPD}$ derivation can be translated into a bottom-up
  $\name_{\instPD}$ derivation.
\end{lemma}
\begin{proof}
  By induction on the shape of the derivation.
  Given that addition is a meet, it is clear that any non-bottom-up uses of
  addition come from one of the arguments being greater than the result.
  Therefore, it is safe to make this argument smaller in the corresponding
  premise (via subusaging), before translating that subderivation.
  For multiplication, again, there is always a lesser value of the right
  argument that will take us from a non-bottom-up fact to a bottom-up fact with
  the same left argument and result.
\end{proof}

\begin{proposition}[\name{} $\to$ PD]
  Given a $\name_{\instPD}$ derivation of
  $\ctx{\Gamma}{\vct\valid}, \ctx{\Delta}{\vct\true}, \ctx{\Theta}{\vct\unused}
  \vdash M : A$ which does not contain types using $\excl{\unused}{}$ or
  $\excl{\true}{}$, we can produce a PD derivation of
  $\Gamma; \Delta \vdash A~\mathit{true}$.
\end{proposition}
\begin{proof}
  We translate away tensor products and tensor units using
  \autoref{lem:top-meet}, and translate the resulting derivation to bottom-up
  form.
  The proof proceeds by induction on the resulting derivation in the obvious
  way.
\end{proof}

For similar reasons as explained at the end of \autoref{sec:syntax}, the
system of Abel and Bernardy \cite{AbelBernardy2020} is unable to embed PD in
this way, as it would prove
$\Box(A \wedge B) \to \Box A \wedge \Box B$, where PD and $\name$ do not.
In fact, this example shows that, even when weakening and contraction are
admissible, with- and tensor-products are distinct in their system in the
presence of modalities.

\section{Metatheory}\label{sec:metatheory}

The motivating result of the following section is the admissibility of the
single substitution principle.
We derive single substitution (\autoref{cor:single-subst}) from simultaneous
substitution (\autoref{cor:sub}), which we in turn derive from a generic
traversal of syntax (\autoref{thm:trav}).
Our statement of single substitution is somewhat standard, easy to intuit, and
exactly what is called for to show preservation of typing by $\beta$ rules,
whereas our statement of simultaneous substitution is novel and more abstract,
but easier to work with.

\begin{mathpar}
  \inferrule*[right=SingleSubst]
  {M : \ctx\Gamma P \vdash A
    \\ N : \ctx\Gamma Q, \ctxvar x A r \vdash B
    \\ \resctx R \subres \rescomment r \resctx P + \resctx Q
  }
  {N\{M\} : \ctx\Gamma R \vdash B}
  \and
  \inferrule*[right=Subst]
  {
    \rho : \ctx\Gamma P \subst\vdash \ctx\Delta Q
    \\
    N : \ctx\Delta Q \vdash A
  }
  {N\{\rho\} : \ctx\Gamma P \vdash A}
\end{mathpar}

The notation $\ctx\Gamma P \subst\vdash \ctx\Delta Q$, which can be thought of
as the type of linear assignments from variables in $\ctx\Delta Q$ to terms in
$\ctx\Gamma P$, will be defined in the following section, and this definition
can be seen as a key contribution of this paper.
The overloaded notations $N\{M\}$ and $N\{\rho\}$ are hereby defined as
operations on intrinsically typed and usage-checked syntax.
Specifically, $N\{M\}$ denotes the term $N$ but with $M$ substituted in for the
most recently bound variable, whereas $N\{\rho\}$ denotes the result of a
generic traversal of $N$ in which its free variables have been replaced
according to the environment $\rho$.

McBride defines \emph{kits} \cite{rensub05,bhkm12}, which provide a general
method for giving admissible rules that are usually proven by induction on the
derivation.
To produce a kit, we give an indexed family
$\kitrel : \mathrm{Ctx} \times \mathrm{Ty} \to \mathrm{Set}$ and explain how to
inject variables, extract terms, and weaken by new variables coming from
binders.
In return, given a type-preserving map from variables in one context to
$\kitrel$-stuff in another (an \emph{environment}), we get a type-preserving
function between terms in these contexts.
Such a function is the intrinsic typing equivalent of an admissible rule.

To make the kit-based approach work in our usage-constrained setting, we make
modifications to both kits and environments.
Kits need not support arbitrary weakening, but only weakening by the
introduction of $\rescomment 0$-use variables.
The family $\kitrel$ must also respect $\subres$ of usage contexts.
Environments are equipped with a matrix mapping input usages to output usages.

We prove simultaneous substitution using renaming.
We take both renaming and substitution as corollaries of the \emph{traversal}
principle (\autoref{thm:trav}) yielded from kits and environments.

Throughout this section, we give definitions, lemmas, and proofs corresponding
directly to parts of the Agda mechanisation.
Agda type formers are highlighted \AgdaRecord{blue}, with other constructions
highlighted \AgdaFunction{green}.
In the PDF version of this paper, all Agda names are hyperlinked to a line of
the source code hosted on GitHub.


\subsection{Kits, Environments, and Traversal}

Agda definitions of \emph{kits}, \emph{environments}, and
\hyperref[thm:trav]{traversal} are in the module
\href{https://github.com/laMudri/generic-lr/blob/lin/tlla-submission-2021/src/Specific/Syntax/Traversal.agda\#L9}{\AgdaModule{Specific.Syntax.Traversal}}
and are explained in this subsection.

\paragraph{Kits}
A kit is a structure on $\vdash$-like relations $\kitrel$, intuitively
giving a way in which $\kitrel$ lives between the usage-checked variable
judgement $\lvar$ and the typing judgement $\vdash$.
The components $\mathit{vr}$ and $\mathit{tm}$ are basically unchanged from
McBride's original kits.
The component $\mathit{wk}$ only differs in that new variables are given
annotation $\rescomment 0$, which intuitively marks them as weakenable.
The requirement $\mathit{psh}$ is new, and allows us to fix up usage contexts
via skew algebraic reasoning.

\begin{definition}[\href{https://github.com/laMudri/generic-lr/blob/lin/tlla-submission-2021/src/Specific/Syntax/Traversal.agda\#L60}{\AgdaRecord{Kit}}]\label{def:kit}
  For any $\kitrel : \mathrm{Ctx} \times \mathrm{Ty} \to \mathrm{Set}$, let
  $\kit(\kitrel)$ denote the type of \emph{kits}.
  A kit comprises the following functions for all $\resctx P$, $\resctx Q$,
  $\Gamma$, $\Delta$, and $A$.
  \begin{mathpar}
    \mathit{psh} : \resctx P \subres \resctx Q \to
    \ctx{\Gamma}{Q} \kitrel A \to \ctx{\Gamma}{P} \kitrel A
    \and
    \mathit{vr} : \ctx{\Gamma}{P} \lvar A \to
    \ctx{\Gamma}{P} \kitrel A
    \\
    \mathit{tm} : \ctx{\Gamma}{P} \kitrel A \to
    \ctx{\Gamma}{P} \vdash A
    \and
    \mathit{wk} : \ctx{\Gamma}{P} \kitrel A \to
    \ctx{\Gamma}{P}, \ctx{\Delta}{\vct 0} \kitrel A
  \end{mathpar}
\end{definition}

An inhabitant of $\ctx{\Gamma}{P} \kitrel A$ is described as
\emph{stuff in $\ctx{\Gamma}{P}$ of type $A$}.
In a traversal (for example, simultaneous renaming and simultaneous
substitution), we use $\mathit{tm}$ in \TirName{var} cases to convert stuff
that has come from the environment into terms that replace variables.
We use $\mathit{vr}$ and $\mathit{wk}$ when binding new variables: $\mathit{vr}$
tells us what to add to the environment when it is being extended by a bound
variable, and $\mathit{wk}$ allows us to weaken all the other stuff in the
environment by such newly bound variables.

\paragraph{Environments}
In simple intuitionistic type theory, an
environment is a type-preserving function from variables in the old
context $\Delta$ to stuff in the new context $\Gamma$: an inhabitant
of $\Delta \ni A \to \Gamma \kitrel A$.  The traversal function turns
such an environment into a map between terms,
$\Delta \vdash A \to \Gamma \vdash A$.

For \name{}, we want maps of usaged terms
$\ctx{\Delta}{Q} \vdash A \to \ctx{\Gamma}{P} \vdash A$.
We can see that an environment of type
$\ctx{\Delta}{Q} \lvar A \to \ctx{\Gamma}{P} \kitrel A$ would
be insufficient --- $\ctx{\Delta}{Q} \lvar A$ can only be inhabited when
$\resctx Q$ is compatible with a basis vector, so our environment would be
trivial in more general cases.
Instead, we care about non-usage-checked variables $\Delta \ni A$.

Our understanding of an environment is that it should simultaneously
map all of the usage-checked variables in $\ctx{\Delta}{Q}$ to stuff
in $\ctx{\Gamma}{P}$ in a way that preserves usage.  As such, we want
to map each variable $j : \Delta \ni A$ not to $A$-stuff in
$\ctx{\Gamma}{P}$, but rather $A$-stuff in $\resctx P_j\Gamma$, where
$\resctx P_j$ is some fragment of $\resctx P$.  Precisely, when
weighted by $\resctx Q\lvert j \rangle$, we want these $\resctx P_j$
to sum to $\resctx P$, so as to provide ``enough'' usage to cover all
of the variables $j$.  When we collect all of the $\resctx P_j$ into a
matrix $\rescomment\Psi$, we notice that the condition just described
is stated succinctly via a vector-matrix multiplication
$\resctx Q\rescomment\Psi$. This culminates to give us the following:

\begin{definition}[\href{https://github.com/laMudri/generic-lr/blob/lin/tlla-submission-2021/src/Specific/Syntax/Traversal.agda\#L43}{\AgdaRecord{Env}}]\label{def:env}
  For any $\kitrel$, $\resctx P$, $\resctx Q$, $\Gamma$, and $\Delta$,
  where $\Gamma$ and $\Delta$ have lengths $m$ and $n$ respectively,
  let $\ctx{\Gamma}{P} \subst{\kitrel} \ctx{\Delta}{Q}$ denote the
  type of \emph{environments}.  An environment comprises a pair of a
  matrix $\rescomment\Psi : \mathbf R^{n \times m}$ and a mapping of
  variables
  $\mathit{act} : (j : (\Delta \ni A)) \to (\langle j
  \rvert\rescomment\Psi)\Gamma \kitrel A$, such that
  $\resctx P \subres \resctx Q \rescomment\Psi$.
\end{definition}

Our main result is the following, which we will instantiate to prove
admissibility of renaming (\autoref{cor:ren}), subusaging
(\autoref{cor:subusage}), and substitution (\autoref{cor:sub}). The
proof is in \autoref{sec:proof-of-traversal}.

\newcommand{\thmtrav}{%
  Given a kit on $\kitrel$ and an environment
  $\ctx{\Gamma}{P} \subst{\kitrel} \ctx{\Delta}{Q}$, we get a function
  $\ctx{\Delta}{Q} \vdash A \to \ctx{\Gamma}{P} \vdash A$.%
}
\begin{theorem}[traversal, \href{https://github.com/laMudri/generic-lr/blob/lin/tlla-submission-2021/src/Specific/Syntax/Traversal.agda\#L83}{\AgdaFunction{trav}}]\label{thm:trav}
  \thmtrav
\end{theorem}

\subsection{Renaming}

We now show how to use traversals to prove that renaming (including
weakening) and subusaging are admissible.
This subsection corresponds to the Agda modules
\href{https://github.com/laMudri/generic-lr/blob/lin/tlla-submission-2021/src/Specific/Syntax/Renaming.agda\#L9}{\AgdaModule{Specific.Syntax.Renaming}}
and
\href{https://github.com/laMudri/generic-lr/blob/lin/tlla-submission-2021/src/Specific/Syntax/Subuse.agda\#L9}{\AgdaModule{Specific.Syntax.Subuse}}.

\begin{definition}[\href{https://github.com/laMudri/generic-lr/blob/lin/tlla-submission-2021/src/Specific/Syntax/Renaming.agda\#L54}{\AgdaFunction{LVar-kit}}]\label{def:lvar-kit}
  Let $\lvar\textrm{-kit} : \kit(\lvar)$ be defined with the following
  fields.
  \begin{description}
    \item[$\mathit{psh}~(\mathit{PQ} : \resctx P \subres \resctx Q)
      : \ctx{\Gamma}{Q} \lvar A \to \ctx{\Gamma}{P} \lvar A$:]
      The only occurrence of the usage context $\resctx Q$ in the definition of
      $\lvar$ is to the left of a $\subres$.
      Applying transitivity in this place gets us the required term.
    \item[$\mathit{vr} : \ctx{\Gamma}{P} \lvar A \to \ctx{\Gamma}{P} \lvar A
      \coloneqq \mathrm{id}$].
    \item[$\mathit{tm} : \ctx{\Gamma}{P} \lvar A \to \ctx{\Gamma}{P} \vdash A
      \coloneqq \TirName{var}$].
    \item[$\mathit{wk} : \ctx{\Gamma}{P} \lvar A
      \to \ctx{\Gamma}{P}, \ctx{\Delta}{\vct 0} \lvar A$:]
      A basis vector extended by $\rescomment 0$s is still a basis
      vector: if that we have $\resctx P \subres \langle i \rvert$ for some $i$,
      we also have
      $\resctx P, \rescomment{\vct 0} \subres \langle \inl i \rvert$.
  \end{description}
\end{definition}

Environments for renamings are special in that the matrix $\rescomment\Psi$ can
be calculated from the action of the renaming on non-usage-checked variables.

\begin{lemma}[\href{https://github.com/laMudri/generic-lr/blob/lin/tlla-submission-2021/src/Specific/Syntax/Renaming.agda\#L62}{\AgdaFunction{ren-env}}]\label{lem:ren-env}
  Given a type-preserving mapping of plain variables
  $f : \Delta \ni A \to \Gamma \ni A$ such that
  $\resctx P \subres \resctx Q\rescomment I_{f\times\id}$,
  we can produce a $\lvar$-environment of type
  $\ctx{\Gamma}{P} \subst{\lvar} \ctx{\Delta}{Q}$.
\end{lemma}
\begin{proof}
  The environment has $\rescomment\Psi \coloneqq \rescomment I_{f\times\id}$,
  so the usage condition holds by assumption.
  Now, $\mathit{act}$ is required to have type
  $(j : \Delta \ni A) \to (\langle j \rvert\rescomment\Psi)\Gamma \lvar A$.
  Take arbitrary $j : \Delta \ni A$.
  Then, we have $f~j : \Gamma \ni A$, so all that is left is to show that $f~j$
  forms a usage-checked variable of type
  $(\langle j \rvert\rescomment\Psi)\Gamma \lvar A$.
  This amounts to proving
  $\langle j \rvert\rescomment\Psi \subres \langle f~j \rvert$.
  Let $i : \Gamma \ni A$, then we have
    $(\langle j \rvert\rescomment\Psi)_i
    \subres \rescomment\Psi_{j,i}
    = \rescomment I_{f\,j,i}
    = \langle f~j \rvert_i$.
\end{proof}

\begin{corollary}[renaming, \href{https://github.com/laMudri/generic-lr/blob/lin/tlla-submission-2021/src/Specific/Syntax/Renaming.agda\#L67}{\AgdaFunction{ren}}]\label{cor:ren}
  Given a type-preserving mapping of plain variables
  $f : \Delta \ni A \to \Gamma \ni A$ such that
  $\resctx P \subres \resctx Q\rescomment I_{f\times\id}$,
  we can produce a function of type
  $\ctx{\Delta}{Q} \vdash A \to \ctx{\Gamma}{P} \vdash A$.
\end{corollary}

\begin{corollary}[subusaging, \href{https://github.com/laMudri/generic-lr/blob/lin/tlla-submission-2021/src/Specific/Syntax/Subuse.agda\#L50}{\AgdaFunction{subuse}}]\label{cor:subusage}
  Given $\resctx P \subres \resctx Q$, then we have a function
  $\ctx{\Gamma}{Q} \vdash A \to \ctx{\Gamma}{P} \vdash A$.
\end{corollary}

\subsection{Substitution}

Now that we have renaming, we can use it with traversals to prove that
simultaneous well usaged substitution is admissible. This subsection
corresponds to the Agda module
\href{https://github.com/laMudri/generic-lr/blob/lin/tlla-submission-2021/src/Specific/Syntax/Substitution.agda\#L9}{\AgdaModule{Specific.Syntax.Substitution}}.

\begin{definition}[\href{https://github.com/laMudri/generic-lr/blob/lin/tlla-submission-2021/src/Specific/Syntax/Substitution.agda\#L63}{\AgdaFunction{Tm-kit}}]\label{tm-kit}
  Let $\vdash\textrm{-kit} : \kit(\vdash)$ be defined with the following
  fields.
  \begin{description}
    \item[$\mathit{psh}~(\mathit{PQ} : \resctx P \subres \resctx Q)
      : \ctx{\Gamma}{Q} \vdash A \to \ctx{\Gamma}{P} \vdash A$:]
      This is \hyperref[cor:subusage]{Corollary \ref*{cor:subusage} (subusaging)}.
      

    \item[$\mathit{vr} : \ctx{\Gamma}{P} \lvar A \to \ctx{\Gamma}{P} \vdash A
      \coloneqq \TirName{var}$].
    \item[$\mathit{tm} : \ctx{\Gamma}{P} \vdash A \to \ctx{\Gamma}{P} \vdash A
      \coloneqq \mathrm{id}$].
    \item[$\mathit{wk} : \ctx{\Gamma}{P} \vdash A \to \ctx{\Gamma}{P},
      \ctx{\Delta}{\vct 0} \vdash A$:] We use \hyperref[cor:ren]{Corollary \ref*{cor:ren} (renaming)}, with
      $f : \Gamma \ni A \to \Gamma, \Delta \ni A$ being the embedding
      $\inl$.  It remains to check that
      $(\resctx P, \rescomment{\vct 0}) \subres \resctx P\rescomment
      I_{\inl\times\id}$.  We prove this pointwise.  Let
      $i : \Gamma, \Delta \ni A$, and take cases on whether $i$ is
      from $\Gamma$ or from $\Delta$.  If $i = \inl i'$ for an
      $i' : \Gamma \ni A$, we must show that
      $\resctx P_{i'} \subres (\resctx P\rescomment
      I_{\inl\times\id})_{\inl i'}$.  But we have the following.
      \[
      \resctx P_{i'} \subres (\resctx P\rescomment I)_{i'}
      = \sum_{j : \Gamma \ni A} \resctx P_j\rescomment I_{j,i'}
      = \sum_{j : \Gamma \ni A} \resctx P_j\rescomment I_{\inl j,\inl i'}
      = (\resctx P\rescomment I_{\inl\times\id})_{\inl i'}.
      \]
      If $i = \inr i'$ for an $i' : \Delta \ni A$, we must show that
      $\rescomment 0 \subres
      (\resctx P\rescomment I_{\inl\times\id})_{\inr i'}$.
      But we have the following.
      \[
      \rescomment 0 \subres (\resctx P\rescomment{\vct 0})_{i'}
      = \sum_{j : \Gamma \ni A} \resctx P_j\rescomment{\vct 0}_{j,i'}
      = \sum_{j : \Gamma \ni A} \resctx P_j\rescomment I_{\inl j,\inr i'}
      = (\resctx P\rescomment I_{\inl\times\id})_{\inr i'}.
      \]
  \end{description}
\end{definition}

We define a simultaneous substitution as an environment of terms.
Expanding definitions, this means that a simultaneous substitution from
$\ctx\Gamma P$ to $\ctx\Delta Q$ is a matrix $\rescomment\Psi$ such that
$\resctx P \subres \resctx Q \rescomment\Psi$, and for each variable $j$ of
type $A$ in $\Delta$, a term $(\langle j \rvert\rescomment\Psi)\Gamma \vdash A$.

\begin{corollary}[substitution, \href{https://github.com/laMudri/generic-lr/blob/lin/tlla-submission-2021/src/Specific/Syntax/Substitution.agda\#L69}{\AgdaFunction{sub}}]\label{cor:sub}
  Given an environment of type
  $\ctx{\Gamma}{P} \subst{\vdash} \ctx{\Delta}{Q}$ (i.e., a
  well usaged simultaneous substitution), we get a function of type
  $\ctx{\Delta}{Q} \vdash A \to \ctx{\Gamma}{P} \vdash A$.
\end{corollary}

\subsection{Single Substitution}

\begin{corollary}[single substitution]\label{cor:single-subst}
  Given $\resctx R \subres \rescomment r \resctx P + \resctx Q$ and terms
  $M : \ctx\Gamma P \vdash A$ and $N : \ctx\Gamma Q, \ctxvar x A r \vdash B$,
  we can produce a term deriving $\ctx\Gamma R \vdash B$.
\end{corollary}
\begin{proof}
  By traversal (specifically, simultaneous substitution, \autoref{cor:sub})
  on $N$.
  We must produce an environment of type
  $\ctx\Gamma R \subst{\vdash} \ctx\Gamma Q, \ctxvar x A r$.
  Let \(
    \rescomment\Psi \coloneqq \left(\begin{array}{c}
                               \rescomment{\mat I}
                               \\ \hline
                               \resctx P
                             \end{array}\right)
  \)
  and notice that $(\resctx Q, \rescomment r)\rescomment\Psi =
  \resctx Q + \rescomment r \resctx P$, so our inequality assumption is enough
  to prove the inequality requirement of environments.
  For the terms to substitute in, we choose the first $\lvert\Gamma\rvert$ terms
  to be their respective variables, and the last term to be $M$.
\end{proof}

\subsection{Proof of Traversal}
\label{sec:proof-of-traversal}

The proof of the traversal theorem follows the same structure as in
McBride's article, extended with proof obligations to show that we are
correctly respecting the usage annotations. We must first prove a
lemma that shows that environments can be pushed under binders.

\begin{lemma}[bind, \href{https://github.com/laMudri/generic-lr/blob/lin/tlla-submission-2021/src/Specific/Syntax/Traversal.agda\#L70}{\AgdaFunction{bindEnv}}]\label{lem:bind}
  Given a kit on $\kitrel$, we can extend an environment of type
  $\ctx{\Gamma}{P} \subst{\kitrel} \ctx{\Delta}{Q}$, to an environment of type
  $\ctx{\Gamma}{P}, \ctx{\Theta}{R} \subst{\kitrel}
  \ctx{\Delta}{Q}, \ctx{\Theta}{R}$.
\end{lemma}
\begin{proof}
  Let the environment we are given be
  $(\rescomment\Psi : \mathbf R^{n \times m},
  \mathit{act} : (j : \Delta \ni A) \to (\langle j \rvert\rescomment\Psi)\Gamma \kitrel A)$,
  with $\resctx P \subres \resctx Q \rescomment\Psi$.
  We are trying to construct
  $(\rescomment{\Psi'} : \mathbf R^{(n + o) \times (m + o)},
  \mathit{act'} : (j : \Delta, \Theta \ni A) \to
  (\langle j \rvert\rescomment{\Psi'})(\Gamma, \Theta) \kitrel A)$,
  with $\resctx P, \resctx R \subres (\resctx Q, \resctx R) \rescomment{\Psi'}$.
  Let \(
    \rescomment{\Psi'} \coloneqq \left(\begin{array}{c|c}
                                  \rescomment\Psi & \rescomment{\mat 0}
                                  \\ \hline
                                  \rescomment{\mat 0} & \rescomment{\mat I}
                                \end{array}\right).
  \)
  With this definition, our required condition splits into the easily checked
  conditions
  $\resctx P \subres \resctx Q\rescomment\Psi + \resctx R\rescomment{\mat 0}$
  and
  $\resctx R \subres
  \resctx Q\rescomment{\mat 0} + \resctx R\rescomment{\mat I}$.
  For $\mathit{act'}$, we take cases on whether $j$ is from $\Delta$ or from
  $\Theta$.
  In the $\Delta$ case, $\mathit{act}$ gets us an inhabitant of
  $(\langle j \rvert\rescomment\Psi)\Gamma \kitrel A$.
  Notice that
  $\langle j \rvert\rescomment{\Psi'} =
  \langle j \rvert\rescomment\Psi, \rescomment{\vct 0}$,
  so we want to get from $(\langle j \rvert\rescomment\Psi)\Gamma \kitrel A$ to
  $(\langle j \rvert\rescomment\Psi)\Gamma, \rescomment{\vct 0}\Theta
  \kitrel A$.
  We can get this using $\mathit{wk}$ from our kit.
  In the $\Theta$ case, notice that
  $\langle j \rvert\rescomment{\Psi'} = \rescomment{\vct 0}, \langle j \rvert$.
  In other words, $\langle j \rvert\rescomment{\Psi'}$ is a basis vector, so we
  actually have usage-checked
  $(\langle j \rvert\rescomment{\Psi'})(\Gamma, \Theta) \lvar A$.
  Thus, we can use $\mathit{vr}$ from our kit to get
  $(\langle j \rvert\rescomment{\Psi'})(\Gamma, \Theta) \kitrel A$, as required.
\end{proof}

\newtheorem*{thm:trav}{\autoref{thm:trav}}
\begin{thm:trav}[traversal, \href{https://github.com/laMudri/generic-lr/blob/lin/tlla-submission-2021/src/Specific/Syntax/Traversal.agda\#L83}{\AgdaFunction{trav}}]
  \thmtrav
\end{thm:trav}
\begin{proof}
  By induction on the syntax of $M$. In the \TirName{var} $x$ case,
  where $x : \ctx{\Delta}{Q} \lvar A$: By definition of $\lvar$, we
  have that $\resctx Q \subres \langle j \rvert$ for some $j$.
  Applying the action of the environment, we have
  $(\langle j \rvert\rescomment\Psi)\Gamma \kitrel A$.  We then have
  $\resctx P \subres \resctx Q\rescomment\Psi \subres \langle j
  \rvert\rescomment\Psi$, so using the fact that stuff appropriately
  respects subusaging ($\mathit{psh}$), we have
  $\ctx{\Gamma}{P} \kitrel A$.  Finally, using $\mathit{tm}$, we get a
  term $\ctx{\Gamma}{P} \vdash A$, as required.

  Non-\TirName{var} cases are generally handled in the following way.
  If the input usage context $\resctx Q$ is split up into a linear
  combination of zero or more usage contexts $\resctx Q_{i}$, obtain a
  similar splitting of $\resctx P$ by setting
  $\resctx P_{i} \coloneqq \resctx Q_{i}\rescomment\Psi$.  This works out
  because of the linearity of matrix multiplication (in particular,
  multiplication respects operations on the left). This yields
  environments of type
  $\resctx P_{i}\Gamma \subst{\kitrel} \resctx Q_{i}\Delta$ for the
  subterms to use with the inductive hypothesis. If any subterms bind
  variables, apply \autoref{lem:bind} as appropriate.
\end{proof}

\section{Conclusion}\label{sec:conclusion}

We have extended McBride's method of kits and traversals to proving
admissibility of renaming, subusaging, and substitution for the
usage-annotated calculus \name{}. In doing so, we have discovered that
only skew semirings are required, and the importance of linear algebra
for stating and proving these results.
We have shown that \name{} is capable of representing several well known
linear and modal type theories by instantiation to different semirings.


As we mentioned in the introduction, there have been several prior
works focused on formalising substructural calculi in proof
assistants. Many of these \cite{power99,XavierORN18,laurent18}
concentrate on sequent calculus presentations of Linear Logic, using
lists to represent contexts and using explicit splitting along permutations
to account for the splitting of contexts in multiplicative rules. In
comparison to \name{}, the use of permutations means that variables in
the context do not have a ``stable name'', which complicates the use
of terms for other purposes, such as assigning a plain semantics that
ignores the linearity. The use of permutations also complicates the
matter of constructing terms within the proof assistant -- explicit
permutations have to be provided at every rule application, making it
difficult to see that a particular term matches an informal named
presentation of a term.

In terms of attaining some level of generality, our work is similar in spirit to the work of Licata, Shulman, and
Riley \cite{LicataSR17}. They give a proof of cut elimination for a
large class of substructural single-conclusion sequent calculi.
The class of natural deduction systems we consider here is less
general, but is not directly comparable.
In particular, we assume contexts form a commutative monoid up to admissible
derivations, whereas Licata, Shulman, and Riley allow contexts to be composed
of arbitrary finitary operators. This allows them to consider, for example, non-commutative systems, which are out of our reach. It would be interesting to see whether the approach of \name{}, and the technique of kits, can be extended beyond unary semiring annotations and to arbitrary $n$-ary operators as in their work.
Our results are also different --- our simultaneous substitution as opposed to
their cut elimination.
We leave a complete comparison to future work.
They have not mechanised their work.

Abel and Bernardy \cite{AbelBernardy2020} have also presented a system
similar to \name{}, along with a relational semantics that allows the
proof of free theorems derived from the usage restrictions imposed by
the chosen semiring. As we mentioned in \autoref{sec:translation},
their system makes some choices that mean it cannot faithfully
represent DILL or PD. Nevertheless, they use our framing of the
metatheory of ``co-effect'' systems in terms of linear algebra, and
the kit technique we have presented here adapts easily to their
setting.

We are currently building on this work to generalise the framework of Allais
\emph{et al.}~\cite{AACMM20} to include usage annotations, allowing
generic metatheory and semantics for an even wider class of
substructural calculi.

\paragraph{Acknowledgements}
We are thankful for comments from Guillaume Allais and Michael Arntzenius.

\bibliographystyle{eptcsalpha}
\bibliography{quantitative}

\end{document}